\documentclass[a4paper,11pt]{article}
\usepackage{a4}
\usepackage{parskip} 
\usepackage{graphicx}    
\usepackage{subfigure}
\usepackage{amsmath, amsthm, amssymb}
\usepackage{texab}
\usepackage{sfspaper}
\usepackage{natbib}

\begin{document}

\newtheorem{theorem}{Theorem}[section]
\newtheorem{proposition}[theorem]{Proposition}

\newtheorem{definition}{Definition}[section]
\newtheorem{lem}{Lemma}[section]

\newcommand{\na}{\texttt{na}} \newtheorem{algo}{Algorithm}[section]

\bibliographystyle{apalike}

\Title{Approximate variances for tapered spectral estimates}

\Author{Michael Amrein and Hans R. K\"unsch}

\maketitle

\begin{abstract}\normalsize
We propose an approximation of the asymptotic 
variance that removes a certain discontinuity in the usual formula for the
raw and the smoothed periodogram in case a data taper is used. It is
based on an approximation of the covariance of the (tapered) periodogram
at two arbitrary frequencies. 
Exact computations of the variances for a Gaussian white noise and an AR(4)
process show that the approximation is more accurate than the usual formula.

\end{abstract}

\section{Introduction}

Spectral estimation is by now a standard topic in time
series analysis, 
and many excellent books are available, e.g. \cite{per} or \cite{blo}. The purpose
of this short note is to propose an approximation of the asymptotic 
variance that removes a certain discontinuity in the usual formula for the
raw and the smoothed periodogram in case a taper is used. The standard
asymptotic variance of the raw periodogram is independent of the 
taper chosen, see Formulae (222b) and (223c)
in \cite{per}. However, this changes when the raw periodogram
is smoothed over frequencies close by. Then a variance inflation
factor $C_h$, see (\ref{eq:C-h}), appears which is equal to one if no taper
is used and greater than one otherwise, compare Table 248 in \cite{per}. 
The reason for this is that tapering introduces correlations between
the raw periodogram at different Fourier frequencies. Because of this,
the variance reduction due to smoothing is smaller in the case of
no tapering. 
The above variance inflation factor is justified asymptotically when
the number of Fourier frequencies that are involved in the smoothing
tends to infinity (more slowly than the number of observations, otherwise
we would have a bias). Hence, if only little smoothing is used, then
we expect something in between: some increase in the variance, but less than
the asymptotic variance inflation factor $C_h$. We
give here a formula, see 
(\ref{eq:var-new}), which is almost as simple as the inflation factor, but
which takes the amount of smoothing into account. 

\section{Notation and preliminaries}
Let $\{X_t\}_{t \in \mathbb{Z}}$ be a real-valued stationary
process with observation frequency $1/\Delta$, mean
$\mathrm{E}[X_t]=\mu$, autocovariances $s_{\tau}:=\mathrm{Cov}(X_t,X_{t+\tau})$
and spectral density $S(f)$. We assume that $X_1, X_2, \ldots, X_N$
have been observed. 
The tapered periodogram (called direct spectral estimator in 
\cite{per}) is 
\begin{equation*}
 \hat{S}^{(tp)}(f) := \frac{\Delta}{\sum_{t=1}^Nh_t^2}\left|
\sum_{t=1}^N h_t(X_t-\tilde{\mu})e^{-i 2\pi i t f \Delta}\right|^2
\end{equation*}
for $f\in[0,1/(2\Delta)]$.
Here the estimator $\tilde{\mu}$ is usually either the arithmetic
mean $\bar{X}$ or the weighted average
$(\sum_{t=1}^N h_t X_t)/(\sum_{t=1}^N h_t).$
The latter has the property that $\hat{S}^{(tp)}(0)=0$. Since 
the choice is irrelevant for the asymptotics, we can use either version. 
The taper $(h_1, \ldots, h_N)$ is chosen to reduce the discontinuities
of the observation window at the edges $t=1$ and $t=N$. Usually,
it has the form
$h_t=h\left((2t-1)/(2N)\right)$
with a function $h$ that is independent of the sample size $N$.
A popular choice is the split cosine taper 
\begin{equation}
\label{split-cosine} 
h^p(x)=\left\{ \begin{array}{ll}
\frac{1}{2}(1-\cos(2 \pi x/p)) & 0 \leq x \leq \frac{p}{2}\\
1                             & \frac{p}{2} < x < 1-\frac{p}{2} \\
\frac{1}{2}(1-\cos(2 \pi(1-x)/p)) & 1-\frac{p}{2} \leq x \leq 1
\end{array} \right\} .
\end{equation}
The tapered periodogram has the approximate variance
\begin{equation}
  \label{eq:var-period}
  \mathrm{Var}[\hat{S}^{(tp)}(f)] \approx S(f)^2, \quad f \notin\{0, 1/(2 \Delta)\}
\end{equation}
(see e.g. \cite{per}, Formula (222b)). In particular, it
does not converge to zero. Because of this, one usually smoothes the
periodogram over a small band of neighboring frequencies. 
We smooth discretely over an equidistant
grid of frequencies. Let $ f_{N',k} = k/(N' \Delta) \ (0 \leq k \leq N'/2)$
for an integer $N'$ of order $O(N)$.
Then the tapered and smoothed spectral estimate is 
\begin{equation*}
\hat{S}^{(ts)}(f_{N',k}) = \sum_{j=-M}^M g_j 
\hat{S}^{(tp)}(f_{N',k-j}),
\end{equation*}
where the $g_j$'s are weights with the properties
$g_j > 0$, $\ g_j=g_{-j}$ $(-M \leq j \leq M)$ and $\sum_{j=-M}^M g_j=1$.
If $k\leq M$, the smoothing includes the value $\hat{S}^{(tp)}(0)$ which
is equal or very close to zero if the mean $\mu$ is estimated. In this
case, we should exclude $j=k$ from the sum. 

\section{Approximations of the variance of spectral estimators}
The usual approximation for the relative variance of
$\hat{S}^{(ts)}(f_{N',k})$ is 
\begin{equation}
  \label{eq:var-smooth}
  \mathrm{Var}\left(\frac{\hat{S}^{(ts)}(f_{N',k})}{S(f_{N',k})}\right) \approx
  C_h \frac{N'}{N}\sum_{r=-M}^M g_r^2 
\end{equation}
for $k \neq 0,N'/2$ where
\begin{equation}
\label{eq:C-h}
C_h = \frac{\sum_{t=1}^{N}h_t^4/N}
{(\sum_{t=1}^N h_t^2/N)^2} .
\end{equation}
This formula is given in \cite{blo}, equation (9.12) on p. 183, and it is
implemented in the 
function ``spec.pgram'' in the language for statistical computing R
(\cite{rr}). In order 
to see that it is the same as Formula (248a) in \cite{per}, one has to go
back to the definition of $W_m$ in terms of the weights $g_j$ which is
given by the formulae (237c), (238d) and (238e).
If we put $M=0$, (\ref{eq:var-smooth}) is different from (\ref{eq:var-period}).
The reason for this difference is that (\ref{eq:var-smooth}) is 
valid in the limit $M \rightarrow \infty$ and $M/N' \rightarrow 0$.
But in applications $M$ is often small, e.g. $M=1$, and one wonders
how good the approximation is in such a case. 

We propose here as alternative the following approximation 
for the relative variance
\begin{equation}
  \label{eq:var-new}
\left(\sum_{r=-M}^M g_r^2 + 2
\sum_{l=1}^{2M}
\frac{\left|H_2^{(N)}(f_{N',l})\right|^2}{H_2^{(N)}(0)^2} \sum_{r=-M}^{M-l}
g_r g_{r+l} \right) 
\end{equation}
(again for $k \neq 0, N'/2$) where
$H_2^{(N)}(f)=\frac{1}{N} \sum_{t=1}^N h_t^2 e^{-i 2 \pi t f \Delta}$.
In order to compute this expression, we need to compute the convolution of
the weights $(g_j)$ and the discrete Fourier transform of 
the squared taper. The former is usually not a problem since $M$ is
substantially smaller than $N'$. Using the fast Fourier transform,
exact computation of the latter is in most cases also possible. 
If not, then by the Lemma below we can use
\begin{equation*} 
H_2^{(N)}(f) \approx \int_0^1 h^2(u) e^{-i 2 \pi N u f \Delta} du \ e^{-i \pi f \Delta}
\frac{\pi f \Delta}{\sin(\pi f \Delta )}.
\end{equation*}
Choosing a simple form for the function $h$, we can compute the
integral on the right exactly. 
It is obvious that (\ref{eq:var-new}) agrees with (\ref{eq:var-period}) 
for $M=0$. In the next section, we show that it also agrees  
with (\ref{eq:var-smooth}) for $M$ large. 

\section{Justification of the approximation}
The idea is simple: We just plug in a suitable approximation for
the relative covariances of the tapered periodogram values into the 
exact expression for the relative variance. 
$\mathrm{Var}\left(\hat{S}^{(ts)}(f_{N',k})/S(f_{N',k})\right)$
is equal to 
\begin{equation}\label{exact1}
\sum_{r=-M}^M \sum_{s=-M}^M g_r g_s 
\mathrm{Cov}\left(\frac{\hat{S}^{(tp)}(f_{N',k-r})}{S(f_{N',k-r})},
\frac{\hat{S}^{(tp)}(f_{N',k-s})}{S(f_{N',k-s})}\right).
\end{equation}
The asymptotic behavior of these covariances is well known.
Theorem 5.2.8 of \cite{bri} shows that, under suitable conditions, we have
for frequencies  
$0 < f \leq g < 1/(2 \Delta)$ that
\begin{equation}
\label{cov-per}
\mathrm{Cov}\left(\frac{\hat{S}^{(tp)}(f)}{S(f)},\frac{\hat{S}^{(tp)}(g)}{S(g)}\right)=\frac{
\left|H_2^{(N)}(f-g)\right|^2+\left|H_2^{(N)}(f+g)\right|^2}
{\left|H_2^{(N)}(0)\right|^{2}}
+O(N^{-1}).
\end{equation}
The statement in \cite{bri} is actually asymmetric in $f$ and $g$ since it
has $S(f)$ instead of $S(g)$ on the left side in the equation above. Our
statement can 
be proved by the same argument if we assume $S(f)\approx S(g)$ when $|f-g|$ is
small. When $|f-g|$ is big, i.e. not of order $O(N^{-1})$, the covariance
is of the order  $O(N^{-1})$ anyhow.
Using the approximation (\ref{cov-per}) directly would
lead to an approximation which  
depends on $k$. Having to compute $N'/2$ different approximate variances
is usually too complicated. However, the term 
$|H_2^{(N)}(f+g)|^2$ is small unless $\Delta(f+g)$ is close to zero modulo 
one. This has been pointed out by \cite{thom}, see also the discussion
on p. 230--231 of \cite{per}. If we omit this term, then
we obtain our new approximation (\ref{eq:var-new}) by a simple
change in the summation indices.

We next give a simple lemma that justifies the omission of the second
term in (\ref{cov-per}). In addition, it also shows how the usual
approximation (\ref{eq:var-smooth}) follows from (\ref{eq:var-new}).

\begin{lem} If $\psi$ is once continuously differentiable on $[0,1]$ and
  $\psi'$ is Lipschitz continuous with constant $L$,
then 
\begin{equation*}
\frac{1}{N}\sum_{t=1}^N \psi \left(\frac{2t-1}{2N}\right)  e^{-i 2 \pi \lambda t}
= \int_0^1 \psi(u) e^{-i 2 \pi N \lambda u} du \ e^{-i \pi \lambda}
\frac{\pi \lambda}{\sin(\pi \lambda)} + R
\end{equation*}
where $ |R| \leq  \mathrm{const.}/N$ uniformly for all $\lambda \in [0,0.5]$.
\end{lem}

\begin{proof}
Put $\epsilon = 1/(2N)$. By a Taylor expansion, we obtain for any $x \in [0,1]$ 
\begin{equation*}
\begin{split}
\int_{x-\epsilon}^{x+\epsilon} \psi(u) e^{-i 2 \pi N \lambda u} du 
&= \psi(x) 
e^{-i 2 \pi N \lambda x }\int_{-\epsilon}^{\epsilon} 
e^{-i 2 \pi N \lambda u} du\\
&+ \psi'(x) e^{-i 2 \pi N \lambda x }\int_{-\epsilon}^{\epsilon} 
u e^{-i 2 \pi N \lambda u} du + R'
\end{split}
\end{equation*}
where the remainder satisfies $|R'| \leq 2\epsilon^3L/3  = L/(12 N^3).$
Next, observe that
$$ \int_{-\epsilon}^{\epsilon} e^{-i 2 \pi N \lambda u} du =
\frac{\sin(\pi \lambda)}{\pi \lambda N},\ 
\int_{-\epsilon}^{\epsilon} u e^{-i 2 \pi N \lambda u} du =
\frac{i}{2 \pi \lambda N^2}\left(\cos(\pi \lambda) -
\frac{\sin(\pi \lambda)}{\pi \lambda} \right).$$
From this the lemma follows by taking $x=(2t-1)/(2N)$ for $t=1,\ldots,N$
and summing up all terms.
\end{proof}

If $\psi(0)=\psi(1)=0$, then by partial integration
$$\left|\int_0^1 \psi(u) e^{-i 2 \pi N \lambda u} du \right|
\leq \frac{\sup |\psi'(x)|}{2\pi\lambda N}.$$ 
Hence by setting $\psi(u)=h^2(u)$, we obtain
\begin{equation}
\label{eq:H}
H_2^{(N)}(f) \leq \textrm{const.} (N f\Delta)^{-1} +
\textrm{const.} N^{-1} \leq \mathrm{const.} (N f\Delta)^{-1}
\end{equation}
for $f\leq 1/\Delta$.
Therefore the second term in (\ref{cov-per}) is negligible unless $f+g$ is
of the order $O(N^{-1})$. 

Finally, we derive the usual variance approximation (\ref{eq:var-smooth})
from (\ref{eq:var-new}) as follows. By Parseval's theorem
\begin{equation*}
\sum_{l=-N'/2}^{N'/2} \left|H_2^{(N)}( f_{N',l})\right|^2=
\frac{N'}{N} \frac{1}{N} \sum_{t=1}^N h_t^4.
\end{equation*}
Note that $H_2^{(N)}(0)=1/N\cdot\sum_{t=1}^N h_t^2$. Because of
(\ref{eq:H}), we have 
$$
\sum_{l=2M+1}^{N'/2} \left|H_2^{(N)}(f_{N',l})\right|^2\leq
\sum_{l=2M+1}^{N'/2} \left(\frac{\mathrm{const.}\cdot
    N'}{Nl\Delta}\right)^2 
\rightarrow 0 
$$
for $M\rightarrow \infty$ and $N'=O(N)\rightarrow \infty$. Thus 
$$\sum_{l=-2M}^{2M} \left|H_2^{(N)}( f_{N',l})\right|^2 -
\frac{N'}{N} \frac{1}{N} \sum_{t=1}^N h_t^4 \rightarrow 0,
$$
also in the above limit. If 
the weights $g_j$ change smoothly as a function $g$ of the lag $j$, i.e.,
$g_j=g(j/M)$, then for any fixed $l$
$$\sum_{r=-M}^{M-l} g_r g_{r+l} \sim M \int_{-1}^1g^2(u)du \sim
\sum_{r=-M}^{M} g_r^2 \ (M \rightarrow \infty)$$
and the desired result follows by dominated convergence.

\section{Comparison with exact relative variances for Gaussian processes} 
If we assume the process  $\{X_t\}_{t \in \mathbb{Z}}$ to be Gaussian, then
it holds
\begin{equation*}
\begin{split}
&\mathrm{Cov}\left(\frac{\hat{S}^{(tp)}(f)}{S(f)},\frac{\hat{S}^{(tp)}(g)}{S(g)}\right)= 
\frac{1}{S(f)S(g)(\sum_{t=1}^Nh_t^2)^2} \\ & \times
\left(\left|\sum_{j,k=1}^Nh_jh_ks_{j-k}e^{-i2\pi(fj-gk)\Delta}\right|^2+\left|\sum_{j,k=1}^Nh_jh_ks_{j-k}e^{-i2\pi(fj+gk)\Delta}\right|^2\right),
\end{split}
\end{equation*}
see p. 326 of \cite{per}. Plugging this into (\ref{exact1}) yields thus an
exact expression. Evaluation is of the order $O(N^3)$, thus it is not
practical to use it routinely. 

We now compare the two approximations to the
exact relative variances for a Gaussian white noise $X_t=\epsilon_t$ and
the AR(4) process  
\begin{equation}\label{ar4def}
X_t=2.7607X_{t-1}-3.8106X_{t-2}+2.6535X_{t-3}-0.9238X_{t-4}+\epsilon_t
\end{equation} used in
\cite{per} (see p. 46) where $\epsilon_t\ \mathrm{i.i.d} \sim
\mathcal{N}(0,1)$. True spectra are shown in Figure \ref{spectra_pic} in
decibel (dB), i.e., the plot displays $10
\log_{10}(S(.))$. 
\begin{figure}[h]
\centering
\includegraphics[scale=.5]{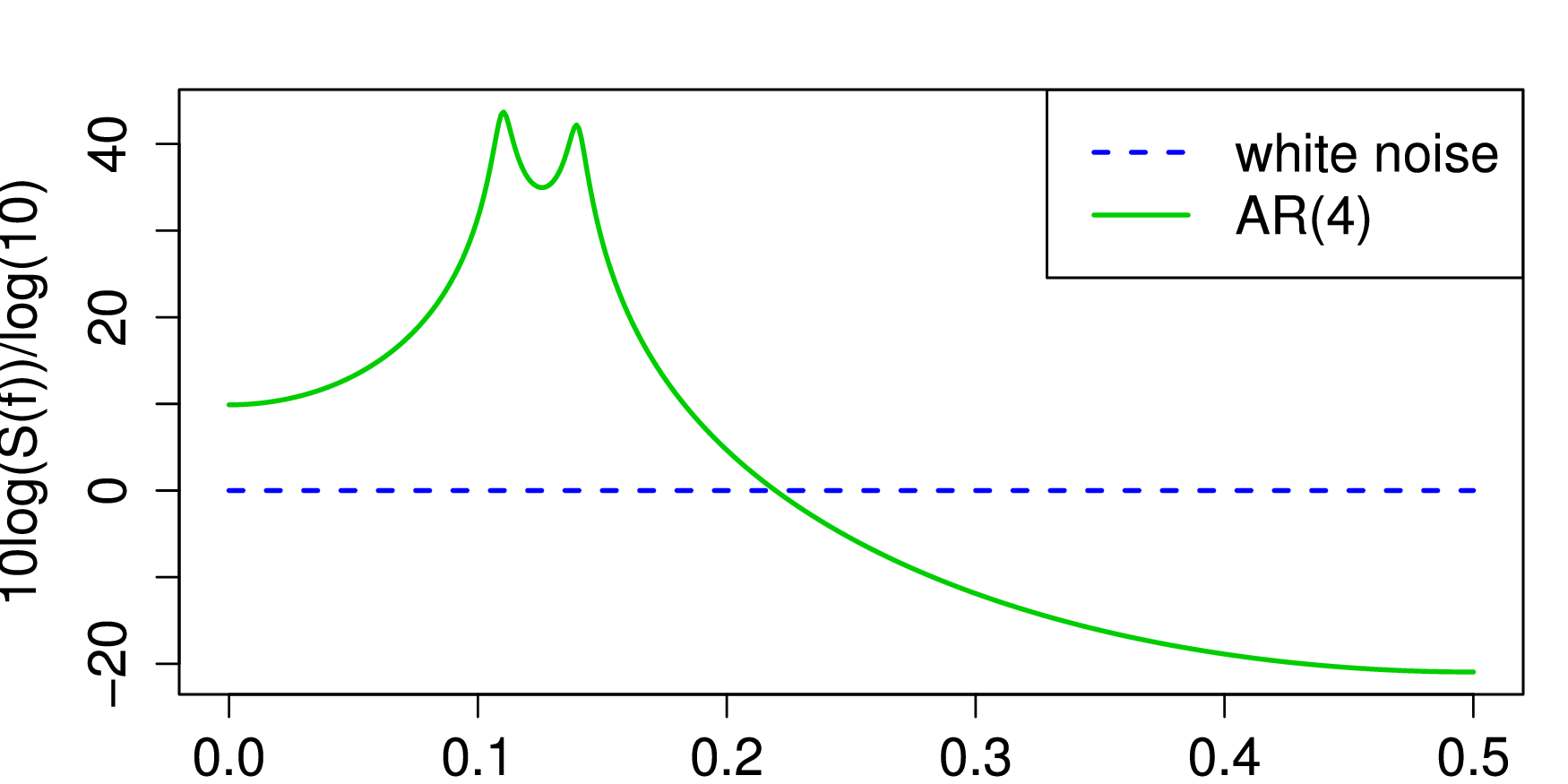}
\caption{Spectra of the Gaussian white noise and the AR(4) process in
  (\ref{ar4def}) in dB.}
 \label{spectra_pic}
\end{figure} 
As we can 
see, the spectrum of the AR(4) process varies over a wide range and exhibits
two sharp peaks. 
Further, we assume the observation frequency $1/\Delta$
to be $1$ and $N=2^{10}=1024$. We compute the exact relative
variance (\ref{exact1}) at the frequencies $f_{k,N'}$, $k=0,\dots,N'/2$, for
$N'\in\{N,2N\}$, 
the split cosine taper (\ref{split-cosine}) with $p\in\{0.2,0.5\}$ and
weights $g_j=1/(2M+1)$, $j=-M,\dots,M$, with $M\in\{0,1,2\}$. Comparison to
the usual approximation (\ref{eq:var-smooth}) and to the new one
(\ref{eq:var-new}) is shown in Figure \ref{ausw_pic}.
The code in R is available under
$\mathbf{http://stat.ethz.ch/\sim kuensch/papers/approximate\_variances.R}$.\\
We see that the new approximation fits the true relative variances clearly
better when we smooth over few frequencies, i.e., $M$ is small. Especially
in the situations when the data is strongly tapered ($p=0.5$) or when we use a
refined smoothing grid ($N'=2N$) we recommend to use the new approximation
(\ref{eq:var-new}). 
\begin{figure*}[h!]
\centering
\includegraphics[scale=.9]{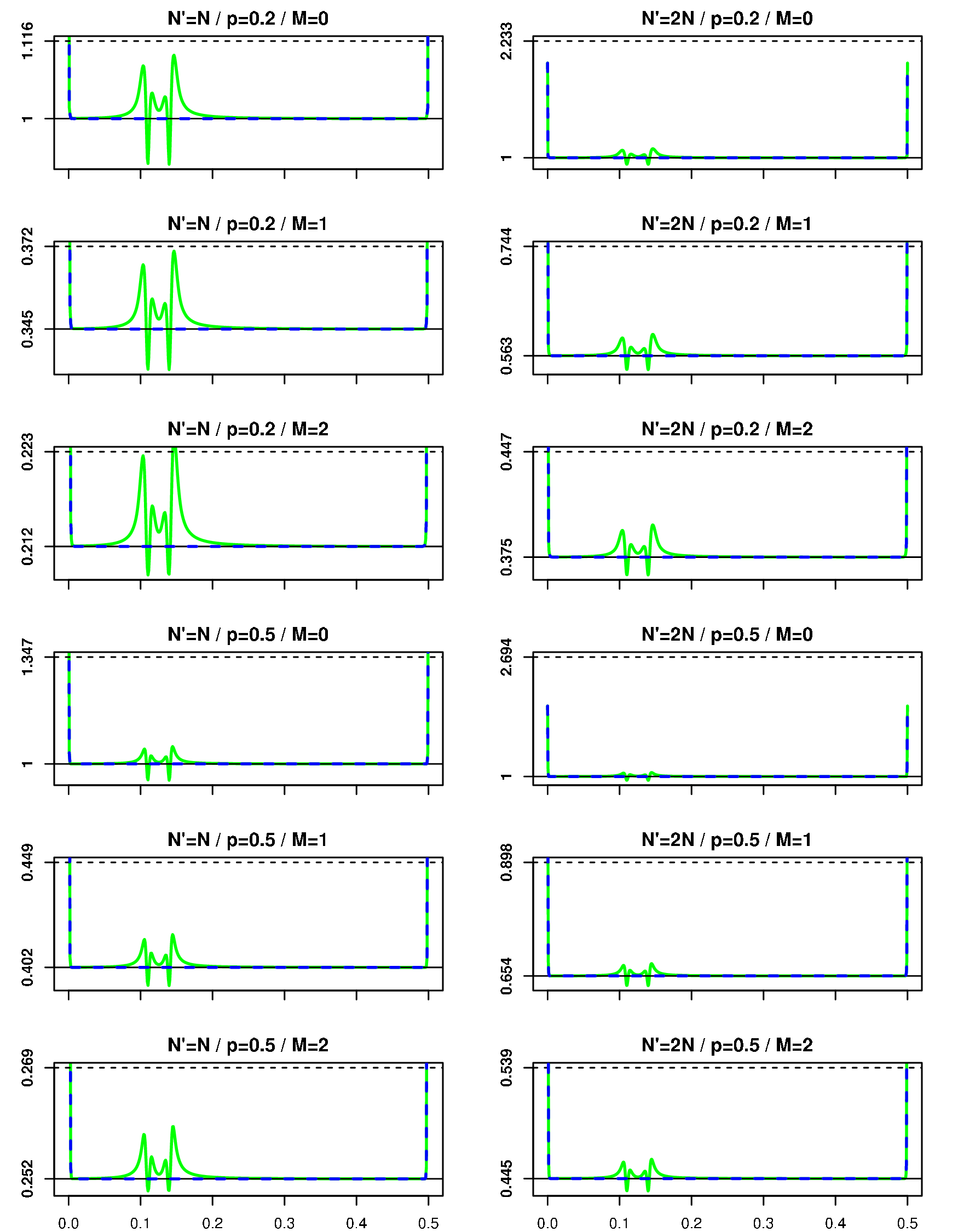}
\caption{Exact relative variances of the \textbf{Gaussian white noise
  (thick-dashed)} and the  
  \textbf{AR(4) process (thick-solid)} in comparison to the \textbf{usual
    (thin-dashed)} and 
  the \textbf{new (thin-solid)} approximation for different choices of
  $N'$, $p$ and $M$. }
 \label{ausw_pic}
\end{figure*} 

\section*{Acknowledgement} We thank Don Percival and Martin M\"achler for
helpful comments and suggestions on earlier versions.

\bibliography{appvar_II}
\end{document}